\documentclass [11pt] {article} 

\usepackage{fullpage}
\usepackage[pdftex]{graphicx}

\usepackage[noend]{algorithmic}

\usepackage{color}
\usepackage{amsmath}
\usepackage{amssymb}
\usepackage{amsfonts}
\usepackage{enumitem}
\usepackage{enumitem}
\usepackage[OT4]{fontenc}

\usepackage[english]{babel}

\usepackage{multirow}
\usepackage{rotating}

\usepackage{epstopdf}

\begin{document}

\newtheorem{theorem}{Theorem}[section]
\newtheorem{lemma}{Lemma}[section]
\newtheorem{corollary}{Corollary}[section]
\newtheorem{claim}{Claim}[section]
\newtheorem{proposition}{Proposition}[section]
\newtheorem{definition}{Definition}[section]
\newtheorem{fact}{Fact}[section]
\newtheorem{example}{Example}[section]

\newcommand{\qed}{\hfill $\square$ \smallbreak}
\newenvironment{proof}[1][Proof]
{\par\noindent{\bf #1:} }{\hspace*{\fill}\nolinebreak{$\Box$}\bigskip\par}

\newcommand{\cV}{\mathcal{V}}
\newcommand{\cG}{\mathcal{G}}
\newcommand{\diam}[1]{\textup{diam}(#1)}
\newcommand{\diff}{\textup{diff}}
\newcommand{\portlabel}{\lambda}
\newcommand{\subdivide}[2]{\xi_{#2}(#1)}
\newcommand{\node}[2]{v_{#1}(#2)}  
\newcommand{\jump}[2]{\textsl{next}_{#2}(#1)}  
\newcommand{\portend}[2]{\textsl{end}_{#2}(#1)}
\newcommand{\modulo}{\,\textup{mod}\,}

\newcommand{\ints}{\mathbb{N}}
\newcommand{\st}{\hspace{0.1cm}\bigl|\bigr.\hspace{0.1cm}}

\title{{\bf Distinguishing Views in Symmetric Networks:\newline A Tight Lower Bound}}

\author{
Dariusz Dereniowski\footnote{Faculty of Electronics, Telecommunications and Informatics, Gda\'{n}sk University of Technology, Poland. E-mail: deren@eti.pg.gda.pl. Partially supported by National Science Centre grant DEC-2011/02/A/ST6/00201.}
\and Adrian Kosowski\footnote{Inria Paris and LIAFA, Universit\'e Paris Diderot, France. E-mail: adrian.kosowski@inria.fr}
\and Dominik Paj\k{a}k\footnote{University of Cambridge, United Kingdom. E-mail: dsp39@cl.cam.ac.uk}
}

\date{}
\maketitle

\thispagestyle{empty}

\begin{abstract}
The view of a node in a port-labeled network is an infinite tree encoding all walks in the network originating from this node.
We prove that for any integers $n\geq D\geq 1$, there exists a port-labeled network with at most $n$ nodes and diameter at most $D$ which contains a pair of nodes whose (infinite) views are different, but whose views truncated to depth $\Omega(D\log (n/D))$ are identical.

\vspace*{10pt}

\end{abstract}

\section{Introduction} \label{sec:intro}
The notion of a \emph{view} was introduced and first studied by Yamashita and Kameda in \cite{YamashitaK96} in the context of distributed message passing algorithms. In so-called anonymous networks (without unique identifiers accessible to a distributed algorithm), the view is a fundamental concept which allows for identification of the network topology and for breaking of symmetries between nodes. 
Different views for a pair of nodes guarantee that the corresponding nodes are distinguishable, which is useful in, e.g., leader election algorithms. View-based approaches have been successfully used when designing algorithms for various network problems, including map construction \cite{CDK10,DP12}, leader election \cite{CMY07,DFNS06,DP13,FP14,TKM12,YK99}, rendezvous \cite{CKP12,DMSVW08,GP11}, and other tasks~\cite{FRS03,YK96}.

The view from a node of a network is by definition (cf. Section~\ref{sec:preliminaries}) an infinite rooted tree, and therefore distributed algorithms (both for agents exploring the network or for the nodes in message passing models) can only know a finite subtree of the view. This motivates the question about the minimum integer $l$ such that the view truncated to depth $l$ contains all crucial information an algorithm may need.

Yamashita and Kameda proved that if views of two nodes truncated to depth $n^2$ are identical, then their infinite views are identical \cite{YamashitaK96}, where $n$ is the number of nodes of the network.
The bound has been improved to $n-1$ by Norris \cite{Norris95}.
Although this bound is asymptotically tight \cite{BoldiV02,Norris95}, it is far from being accurate for many networks.
Hence, one may ask for bounds expressed as function of different graph invariants.
Fraigniaud and Pelc proved in \cite{FraigniaudP12} that if two nodes have the same views do depth $\widehat{n}-1$ then their views are the same, where $\widehat{n}$ is the number of nodes having different views (or equivalently, $\widehat{n}$ is the size of the quotient graph \cite{YamashitaK96}). For some works on view computation see, e.g., \cite{BV02,Tani12}. Recently, Hendrickx~\cite{Hendrickx14} proved (for simple graphs with symmetric port labeling) an upper bound of $O(D\log(n/D))$ on the depth to which views need to be checked in order to be distinguished, where $D$ is the diameter of the network, leaving the tightness of this bound as an open problem.

In this work we provide a corresponding lower bound of $\Omega(D\log(n/D))$. In particular, for each $D'\geq 3$ and $n'\geq D'\cdot 2^{12}/3$, we construct an $n'$-node graph $G'$ with diameter at most $D'$ such that taking truncations of the view to depth $\frac{D'-5}{6}\log_2\frac{n'}{D'} - 0.41D'$ does not guarantee distinguishing a pair of nodes of this graph, which do in fact have different (infinite) views.
Our construction is done in two steps.
First, a list of graphs $G_l$, $l\geq 1$, is defined with the following properties: (a) $\diam{G_l}=3$ for each $l\geq 1$, and (b) $G_l$ contains two nodes $a_l$ and $b_l$ such that the views from them to depth $l=\Theta(\log n)$ are identical but their (infinite) views are different, where $n$ is the size of $G_l$.
Next, in order to extend the bound for arbitrarily large diameter $D'$ we then modify $G_l$ by subdividing each of its edges roughly $D'/3$ times so that the new graph: (a) has diameter roughly $D'$, and (b) contains two nodes $a_l$ and $b_l$ such that their views are the same till depth $\Theta(D'\log_2(n'/D'))$ but their views are different, where $n'$ is the size of the subdivided graph.

{
We remark that very recently~\cite{FP14}, a construction of a class of labeled graphs has been put forward in the context of lower bounds for the leader election problem on anonymous graphs, which can also be used to obtain a separation of node views at distance $\Theta(\log n)$ in a graph of diameter $D = O(1)$. The analysis of that class appears somewhat more involved than for our construction.}

\section{Preliminaries} \label{sec:preliminaries}

In this work we consider anonymous port labeled networks (the terms graph and network are used interchangeably throughout) in which the nodes do not have identifiers and each edge $\{u,v\}$ has two integers assigned to its endpoints, called the \emph{port numbers} at $u$ and $v$, respectively. The port numbers are assigned in such a way that for each node $v$ they are pairwise different and they form a consecutive set of integers $\{1,\ldots,k\}$, where $k$ is the number of neighbors of $v$ in $G$.
The number of neighbors of $v$ in $G$ is called the \emph{degree} of $v$ and is denoted by $\deg_G(v)$.
To simplify some statements we introduce a \emph{port labeling} function $\portlabel$ for $G$ defined in such a way that for each pair $u,v$ of adjacent nodes, $\portlabel(u,v)$ is the port label at $u$ of the edge $\{u,v\}$.
For each node $v$ of $G$ and for each $p\in\{1,\ldots,\deg_G(v)\}$, $\jump{v}{p}$ is the node $u$ such that $\portlabel(v,u)=p$, whereas $\portend{v}{p} = \portlabel(\jump{v}{p},v)$ is the port label at the other end of the edge.

We recall the definition of a view \cite{YamashitaK96}.
Let $G$ be a graph, $v$ be a node of $G$ and let $\portlabel$ be a port labeling for $G$.
Given any $l \geq 0$, the {\em (truncated) view up to level $l$}, $\cV_l(v)$, is defined as follows.
$\cV_0(v)$ is a tree consisting of a single node $x_0$. 
Then, $\cV_{l+1}(v)$ is the port-labeled tree rooted at $x_0$ and constructed as follows.
For every node $v_i$, $i\in\{1,\ldots,\deg_G(v)\}$, adjacent to $v$ in $G$ there is a child $x_i$ of $x_0$ in $\cV_{l+1}(v)$ such that the port number 
at $x_0$ corresponding to edge $\{x_0,x_i\}$ equals $\portlabel(v,v_i)$, and  the port number at $x_i$ corresponding to edge $\{x_0,x_i\}$ equals $\portlabel(v_i,v)$.
For each $i\in\{1,\ldots,\deg_G(v)\}$ the node $x_i$ is the root of the truncated view $\cV_l(v_i)$. 
 
The {\em view} from $v$ in $G$ is the infinite port-labeled rooted tree $\cV(v)$ such that $\cV_l(v)$ is its truncation to level $l$, for each $l\geq 0$.

We remark that by adopting the above definitions, we are considering so-called \emph{symmetric} networks in the sense that the port-labeled network corresponds to an unlabeled graph which is undirected, and  that the encoding of port numbers at both endpoints of each edge appears in the labeling of the edges of the view.

A path in $G$ is denoted as a sequence of nodes, $P=(v_0,v_1,\ldots,v_k)$, such that $\{v_0,\ldots,v_k\}\subseteq V(G)$ and $\{v_i,v_{i+1}\}$ is an edge in $G$ for each $i\in\{0,\ldots,k-1\}$. Note that nodes may repeat in a path, i.e., we do not assume that $v_i \neq v_j$ for $i\neq j$. 
We say that two paths $P_1=(u_0,u_1,\ldots,u_k)$ and $P_2=(v_0,v_1,\ldots,v_k)$ in $G$ are \emph{isomorphic} if $\portlabel(u_i,u_{i+1})=\portlabel(v_i,v_{i+1})$ and $\portlabel(u_{i+1},u_i)=\portlabel(v_{i+1},v_i)$ for each $i\in\{0,\ldots,k-1\}$. We will call a path \emph{non-backtracking}\footnote{Boldi and Vigna used in \cite{BoldiV02} the term ``non-stuttering'' to denote such paths.} if it never follows the same edge twice on end in  opposite directions, i.e., $\portlabel(v_{i},v_{i-1}) \neq \portlabel(v_{i},v_{i+1}) $ for all $i \in \{1,\ldots,k-1\}$.

\begin{claim}[\cite{YamashitaK96}] \label{claim:isomorphic-paths}
Let $G$ be a graph, let $u,v$ be two nodes of $G$, and let $l\geq 0$ be an integer.
We have $\cV_l(u)=\cV_l(v)$ if and only if, for any path of length $l$ starting at $u$, there exists an isomorphic path of length $l$ starting at $v$, and vice versa. The claim also holds when restricting considerations to non-backtracking paths.
\end{claim}

We write $\diam{G}$ to denote the \emph{diameter} of $G$, i.e., the maximum (taken over all pairs of nodes $u$ and $v$) length of a shortest path between $u$ and $v$ in $G$.

\section{The lower bound} \label{sec:G}

For each $l>1$ we define the graph $G_l$ which consists of nodes laid out on a regular grid with $l+2$ \emph{levels} and $2^l$ \emph{columns}, where the node in level $i\in\{0,1,\ldots,l+1\}$ and column $j \in\{0,1,\ldots,2^{l}-1\}$ is denoted by $\node{i}{j}$. Note that all levels are of size $2^l$, and $n_l=|V(G_l)|=(l+2)2^l$.

The construction of the edge set of $G_l$ proceeds in four stages.
Before giving a formal construction, we first provide some intuitions regarding the purpose served by edges introduced in different stages.
The edges added to $G_l$ in Stages~2 and~3 ensure that the graph is connected and has diameter of fixed size.
The aim of Stage~3 is to add edges between consecutive levels in such a way that if one wants to detect a difference between some pairs of nodes in level $l+1$ (e.g., $v_{l+1}(0)$ and $v_{l+1}(2^{l-1})$), then two paths of sufficient length from those nodes need to be selected.
In particular, the paths first need to go through all levels and reach level $0$ (in the mentioned case, these are the nodes $v_0(0)$ and $v_0(1)$).
The edges added to $G_l$ in Stage~1 ensure that nodes in level $0$ from two consecutive columns have different views truncated to depth $2$.

\paragraph{Stage 1. Edges within level $0$.} In level $0$, the edges form a matching between nodes $v_0(j)$ and $v_0(j\oplus 1)$, $j\in\{0,\ldots,2^{l}-1\}$, with ports with labels $\{1,2\}$, given as follows:

\begin{algorithmic}
\FOR {$j :=  0,\ldots,2^{l}-1$}
   \STATE {$\portlabel(\node{0}{j},\node{0}{j\oplus 1}) := 1 + ((j+1)\modulo 2)$;}
\ENDFOR
\end{algorithmic}
In the above, $\oplus$ denotes the xor operation (bitwise modulo-2 addition of non-negative integers).

\paragraph{Stage 2. Edges within level $l+1$.} The edges in level $l+1$ form a clique on all $2^l$ nodes of the level, with port labels corresponding to the difference of identifiers of the connected nodes, computed modulo $2^l$.

\begin{algorithmic}
\FOR {$j := 0,\ldots,2^{l}-1$}
   \FOR {$p := 1,\ldots,2^{l}-1$}
        \STATE {$\portlabel(\node{l+1}{j} , \node{l+1}{(j+p)\modulo 2^{l}}) := p$.}
   \ENDFOR
\ENDFOR
\end{algorithmic}

\paragraph{Stage 3. Edges connecting level $l+1$ with all lower levels.} Each node $v_{l+1}(j)$ from level $l+1$ is connected to all nodes lying in lower levels, in the same column. The port numbers at node $v_{l+1}(j)$ leading to successive levels are successive integers starting from $2^l$, and the port numbers at the other end of such edges are always equal to $1$, except for level $0$, where the port label is either $1$ or $2$ (depending on which port was not used at the considered node in Stage 1 of the construction):

\begin{algorithmic}
\FOR {$j := 0,\ldots,2^{l}-1$}
   \FOR {$i := 0,\ldots,l$}
        \STATE {$\portlabel(\node{l+1}{j} , \node{i}{j}) := 2^l + i $;}
        \IF {$i>0$} 
				\STATE {$\portlabel(\node{i}{j} , \node{l+1}{j}) := 1 $.}
		  \ELSE 
				\STATE {$\portlabel(\node{0}{j} , \node{l+1}{j}) := 1 + (j\modulo 2) $.}
		  \ENDIF
   \ENDFOR
\ENDFOR
\end{algorithmic}

\paragraph{Stage 4. Edges connecting adjacent levels.} Each node belonging to a level $i\in\{0,\ldots, l-1\}$ is connected by an edge to exactly one node of the level $i+1$ directly above, so that the set of edges between such two adjacent levels is a matching. Specifically, we introduce a permutation $\pi_i$ on the set of integers $\{0,\ldots,2^l-1\}$, defined for $i=0$ as the identity permutation $\pi_0(j)=j$, and for $i>0$ as the involution (a function that is its own inverse) which swaps the values of the $i$-th and $(i-1)$-th rightmost bits in the binary notation of its argument:
\begin{equation}\label{eq:involution}
\pi_i(j) = (j -  2^i b_{i}(j) - 2^{i-1} b_{i-1}(j)) + 2^{i} b_{i-1}(j) + 2^{i-1} b_{i}(j),
\end{equation}
where for $k\geq 0$, $b_k(j) = 1$ if $(j \mod 2^{k+1}) \geq 2^{k}$, and $b_k(j) = 0$, otherwise.
For each node at level $i\in \{1,\ldots,l-1\}$, the port label used on the edge leading to level $i-1$ is always $2$, and the port label leading to level $i+1$ is always $3$, as follows:

\begin{algorithmic}
\FOR {$j := 0,\ldots,2^{l}-1$}
   \FOR {$i := 0,\ldots,l-1$}
	  \STATE {$\portlabel(\node{i}{j} , \node{i+1}{\pi_{i}(j)} := 3$;}
     \STATE {$\portlabel(\node{i+1}{\pi_{i}(j)} , \node{i}{j}) := 2$.}
   \ENDFOR
\ENDFOR
\end{algorithmic}

The graph $G_4$ with some edges omitted is shown in Figure~\ref{fig:Gl}.
\begin{figure}[htb]
\centering
\includegraphics[width=0.7\textwidth]{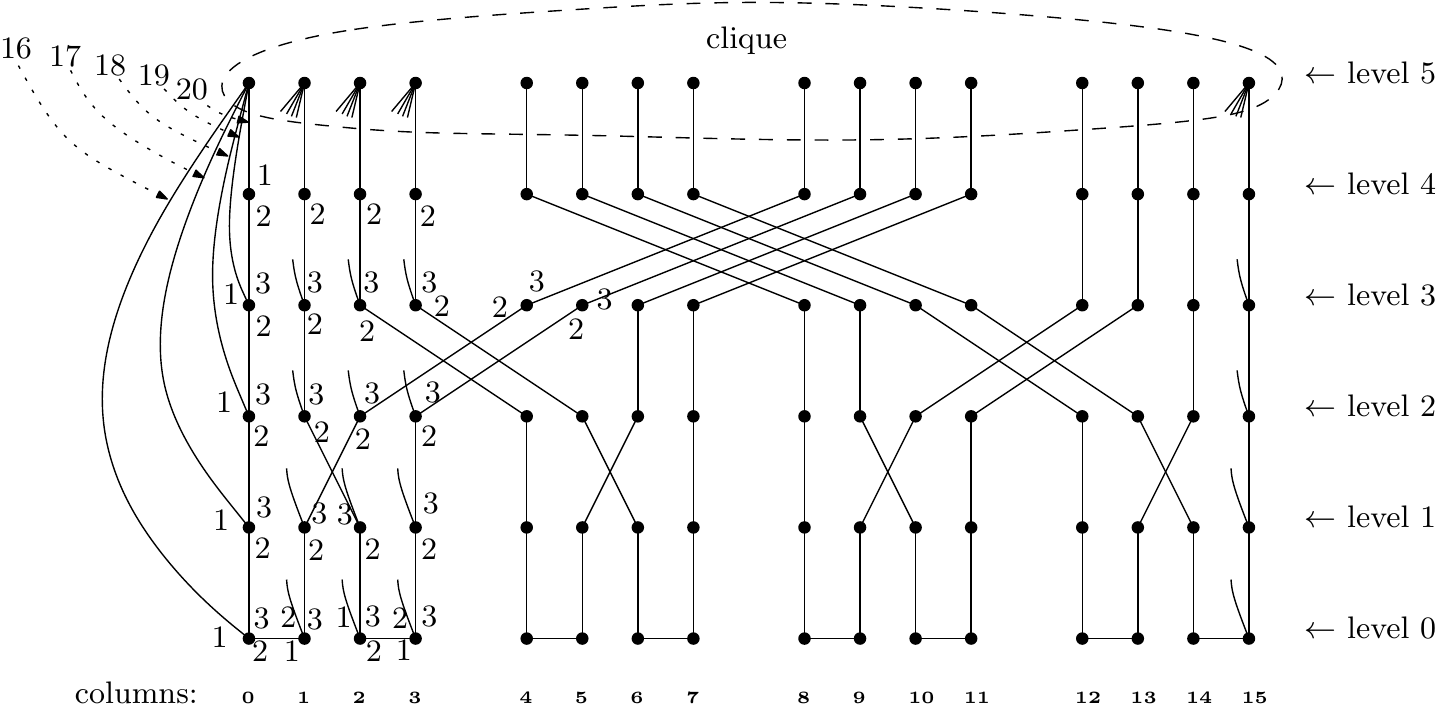}
\caption{The construction of $G_l$ for $l=4$}
\label{fig:Gl}
\end{figure}
In particular, the edges between nodes in level $l+1$ and level $i$, $i\leq l$, are given only in column 0, and edges from the clique in level $l+1$ are omitted.

\begin{claim} \label{claim:diam-of-G_l}
For each $l\geq 6$ it holds that  $|E(G_l)| < 2^{2l}$ and $\diam{G_l}\leq 3$.
\end{claim}
\begin{proof}
The number of edges of $G_l$ can be bounded by counting the number of edges added in Stages 1 to 4 and bounding for $l\geq 6$.

To bound the diameter, note that any node of $G_l$ either belongs to level $l+1$ or is within distance $1$ from a node in level $l+1$. Also, any two nodes in level $l+1$ are adjacent.
\end{proof}

For a pair of integers $0\leq j_1, j_2 < 2^l$, we will denote by $\delta(j_1, j_2)$ the number of rightmost bits in their binary representations which are all identical, i.e., $\delta(j_1, j_2)$  is the largest integer $\delta \in \{0,\ldots,l\}$ such that $(j_1\equiv j_2) \modulo 2^{\delta}$ (or equivalently, such that $b_k(j_1) = b_k(j_2)$ for all $0\leq k <\delta$).  The function $\delta(j_1, j_2)$ has several important properties with respect to transformations of its parameters.

\begin{lemma}\label{lem:delta}
Let $j_1, j_2, d \in \{0,\ldots,2^{l}-1\}$ and $i\in\{1,\ldots,l-1\}$ be arbitrarily chosen. Then:
\begin{enumerate} [label={\normalfont(\roman*)},leftmargin=*]
\item\label{it:delta1} $\delta(j_1 \oplus d, j_2 \oplus d) = \delta(j_1,j_2)$.
\item\label{it:delta2} $\delta((j_1+d)\modulo 2^l, (j_2+d)\modulo 2^l) = \delta(j_1,j_2)$.
\item\label{it:delta3} $\delta(\pi_i (j_1), \pi_i (j_2)) \geq \delta(j_1,j_2) -1$, where involution $\pi_i$ is defined by \eqref{eq:involution}.
\end{enumerate}
\end{lemma}
\begin{proof}
Claims \ref{it:delta1} and \ref{it:delta2} can be attributed to folklore. To prove claim \ref{it:delta3}, note that the involution $\pi_i$ consists in swapping adjacent bits at positions $i$ and $i-1$, only. Consequently, we have by definition of $\delta(j_1,j_2)$ that if $\delta(\pi_i (j_1), \pi_i (j_2)) \neq \delta(j_1,j_2)$ then either $\delta(j_1,j_2) = i-1$ or $\delta(j_1,j_2) = i$. In both cases, we have $\delta(\pi_i (j_1), \pi_i (j_2)) \geq i-1$, and claim \ref{it:delta3} follows.
\end{proof}
In the following, for two nodes $v_i(j_1)$ and $v_i(j_2)$ belonging to the same level $i$ of $G_l$, we will use the notation: $\delta(v_i(j_1),  v_i(j_2)) \equiv \delta(j_1,j_2)$.

\begin{lemma} \label{lem:difference-remains}
Consider a pair of nodes $\node{i}{j_1}$, $\node{i}{j_2}$ of $G_l$ with $\delta(\node{i}{j_1},\node{i}{j_2})>0$. Then:
\begin{enumerate} [label={\normalfont(\roman*)},leftmargin=*]
\item\label{it:diff1} Nodes $\node{i}{j_1}$ and $\node{i}{j_2}$ are of the same degree $d$.
\item\label{it:diff4} For any port $p\in \{1,\ldots,d\}$, nodes $\jump{\node{i}{j_1}}{p}$ and $\jump{\node{i}{j_2}}{p}$ belong to the same level in $G_l$. 
\item\label{it:diff2} For any port $p\in \{1,\ldots,d\}$, $\delta(\jump{\node{i}{j_1}}{p}, \jump{\node{i}{j_2}}{p} \geq \delta(\node{i}{j_1},\node{i}{j_2}) -1$.
\item\label{it:diff3} For any port $p\in \{1,\ldots,d\}$, $\portend{\node{i}{j_1}}{p} = \portend{\node{i}{j_2}}{p}$.
\end{enumerate} 
\end{lemma}
\begin{proof}
By the construction of $G_l$, all nodes in the same level are of the same degree, and claim \ref{it:diff1} follows.
Claim \ref{it:diff4} also follows directly from the construction of $G_l$.

The construction of the port labeling in $G_l$ is such that the Stage $a \in \{1,2,3,4\}$, during which an edge along any port $p$ is added to a vertex $\node{i}{j}$, depends only on the value of its level $i$ and the parity $j \modulo 2$ of its column number (this parity is only relevant for the case of $i=0$ and $p=2$, distinguishing edges added in Stage 1 and Stage 3). The nodes $\node{i}{j_1}$ and $\node{i}{j_2}$ belong to the same level. Moreover, since $\delta(j_1,j_2)>0$, we have that
\begin{equation} \label{eq:equiv}
 (j_1 \equiv j_2)\mod 2.
\end{equation}
It follows that the edges $e_1 = \{\node{i}{j_1}, \jump{\node{i}{j_1}}{p}\}$ and $e_2 = \{\node{i}{j_2}, \jump{\node{i}{j_2}}{p}\}$, corresponding to a traversal of the same port $p$ starting from nodes $\node{i}{j_1}$ and $\node{i}{j_2}$, must necessarily have been defined in the same Stage $a$ of the construction of the edge set of $G_l$. To complete the proofs of claims \ref{it:diff2} and \ref{it:diff3}, we consider the corresponding four cases of $a\in \{1,2,3,4\}$.

\begin{itemize}
\item Edges $e_1$ and $e_2$ were defined in Stage 1. Then, $i=0$, $p\in\{1,2\}$, and we have: 
$$\jump{\node{i}{j_1}}{p} = \node{i}{j_1 \oplus 1}, \quad 
\portend{\node{i}{j_1}}{p} = 1 + (j_1 \modulo 2),$$
$$\jump{\node{i}{j_2}}{p} = \node{i}{j_2 \oplus 1}, \quad 
\portend{\node{i}{j_2}}{p} = 1 + (j_2 \modulo 2).$$
By Lemma~\ref{lem:delta}\ref{it:delta1}, we have:
$$
\delta(\jump{\node{i}{j_1}}{p}, \jump{\node{i}{j_2}}{p} = \delta(\node{i}{j_1},\node{i}{j_2}).
$$
Moreover, taking into account \eqref{eq:equiv}, we obtain $\portend{\node{i}{j_1}}{p} = \portend{\node{i}{j_2}}{p}$. This completes the proof of claims \ref{it:diff2} and \ref{it:diff3} for this case.

\item Edges $e_1$ and $e_2$ were defined in Stage 2. Then, $i=l+1$, $p\in\{1,\ldots,2^l-1\}$, and we have: 
$$\jump{\node{i}{j_1}}{p} = \node{i}{(j_1 +p ) \modulo 2^l}, \quad 
\portend{\node{i}{j_1}}{p} = (2^l-p) \modulo 2^l,$$
$$\jump{\node{i}{j_2}}{p} = \node{i}{(j_2 +p ) \modulo 2^l}, \quad 
\portend{\node{i}{j_2}}{p} = (2^l-p) \modulo 2^l.$$
We immediately have $\portend{\node{i}{j_1}}{p} = \portend{\node{i}{j_2}}{p}$, and moreover, by Lemma~\ref{lem:delta}\ref{it:delta2}:
$$
\delta(\jump{\node{i}{j_1}}{p}, \jump{\node{i}{j_2}}{p} = \delta(\node{i}{j_1},\node{i}{j_2}).
$$

\item Edges $e_1$ and $e_2$ were defined in Stage 3. Then, we need to consider two cases: either $i = l+1$, or $i\in\{0,\ldots l\}$. 

If $i=l+1$, then $p = 2^l + i'$ for some $i' \in \{0,\ldots l\}$. We have for $i'>0$:
$$\jump{\node{i}{j_1}}{p} = \node{i'}{j_1}, \quad 
\portend{\node{i}{j_1}}{p} = 1,$$
$$\jump{\node{i}{j_2}}{p} = \node{i'}{j_2}, \quad 
\portend{\node{i}{j_2}}{p} = 1,$$
whereas for $i'=0$:
$$\jump{\node{i}{j_1}}{p} = \node{0}{j_1}, \quad 
\portend{\node{i}{j_1}}{p} = 1 + (j_1 \modulo 2),$$
$$\jump{\node{i}{j_2}}{p} = \node{0}{j_2}, \quad 
\portend{\node{i}{j_2}}{p} = 1 + (j_2 \modulo 2).$$
Claims \ref{it:diff2} and \ref{it:diff3} follow directly, taking into account Equation~\eqref{eq:equiv} in the latter case.

Otherwise, if $i<l+1$, then $p = 2$ (if $i=0$ and $j_1\equiv j_2\equiv 1 \mod 2$), or $p=1$ (in all other cases). We have:
$$\jump{\node{i}{j_1}}{p} = \node{l+1}{j_1}, \quad 
\portend{\node{i}{j_1}}{p} = 2^l + i,$$
$$\jump{\node{i}{j_2}}{p} = \node{l+1}{j_2}, \quad 
\portend{\node{i}{j_2}}{p} = 2^l + i,$$
and claims \ref{it:diff2} and \ref{it:diff3} immediately follow as well.

\item Edges $e_1$ and $e_2$ were defined in Stage 4. Then, $p\in \{2,3\}$ and $i\in\{0,\ldots,l\}$.

We first consider the case of $p=3$, i.e., when $i < l$ and port $p$ leads up to level $i+1$. We have:
$$\jump{\node{i}{j_1}}{p} = \node{i+1}{\pi_{i}(j_1)}, \quad 
\portend{\node{i}{j_1}}{p} = 2,$$
$$\jump{\node{i}{j_2}}{p} = \node{i+1}{\pi_{i}(j_2)}, \quad 
\portend{\node{i}{j_2}}{p} = 2.$$
Claim \ref{it:diff3} follows directly, and so does claim \ref{it:diff2}, taking into account that by Lemma~\ref{lem:delta}\ref{it:delta3}:
$$
\delta(\jump{\node{i}{j_1}}{p}, \jump{\node{i}{j_2}}{p} = \delta(\pi_{i}(j_1),\pi_{i}(j_2))\geq \delta(j_1,j_2)-1 = \delta(\node{i}{j_1},\node{i}{j_2})-1.
$$

In the case of $p=2$, i.e., when $i > 0$ and port $p$ leads down to level $i-1$, we have:
$$\jump{\node{i}{j_1}}{p} = \node{i-1}{\pi_{i-1}^{-1}(j_1)}, \quad 
\portend{\node{i}{j_1}}{p} = 3,$$
$$\jump{\node{i}{j_2}}{p} = \node{i-1}{\pi_{i-1}^{-1}(j_2)}, \quad 
\portend{\node{i}{j_2}}{p} = 3.$$
We obtain the claims as in the previous case, this time noting that since $\pi_{i-1}$ is an involution, we have $\pi_{i-1}^{-1} \equiv \pi_{i-1}$, and we can apply Lemma~\ref{lem:delta}\ref{it:delta3} for $\pi_{i-1}$ to show Claim \ref{it:diff2}.
\end{itemize}
\vspace{-7mm}
\end{proof}

\begin{lemma} \label{lem:truncated_depth}
Consider a pair of nodes $\node{i}{j_1}$, $\node{i}{j_2}$ of $G_l$ with $\delta \equiv \delta(\node{i}{j_1},\node{i}{j_2}) >0$. Then, the views of nodes $\node{i}{j_1}$ and $\node{i}{j_2}$ are equal at least up to depth $\delta$, $\cV_\delta(\node{i}{j_1}) = \cV_\delta(\node{i}{j_2})$.
\end{lemma}
\begin{proof}
The proof proceeds by induction with respect to $\delta$. 

When $\delta =1$, by Lemma~\ref{lem:difference-remains}\ref{it:diff1}, the nodes $\node{i}{j_1}$  and $\node{i}{j_2}$ have the same degree $d$, and by Lemma~\ref{lem:difference-remains}\ref{it:diff3}, after traversing an edge labeled with any port $p\in\{1,\ldots,d\}$ from either node, we enter the adjacent node by the same port: $\portend{\node{i}{j_1}}{p} = \portend{\node{i}{j_2}}{p}$.
Hence, $\cV_1(\node{i}{j_1}) = \cV_1(\node{i}{j_2})$.

Now, let $\delta > 1$ and suppose that the claim of the lemma holds for all $\delta' \leq \delta -1$. Again, by Lemma~\ref{lem:difference-remains}\ref{it:diff1} and \ref{it:diff3}, the nodes $\node{i}{j_1}$  and $\node{i}{j_2}$ have the same degree, and after traversing an edge labeled with any port $p\in\{1,\ldots,d\}$ from either node, we enter the adjacent node by the same port.
Moreover, we have by Lemma~\ref{lem:difference-remains}\ref{it:diff2} that $\delta(\jump{\node{i}{j_1}}{p}, \jump{\node{i}{j_2}}{p} \geq \delta -1$,
and, by Lemma~\ref{lem:difference-remains}\ref{it:diff4}, $\jump{\node{i}{j_1}}{p}$ and $\jump{\node{i}{j_2}}{p}$ belong to the same level of $G_l$.
Hence, by the inductive assumption, $\cV_{\delta-1}(\jump{\node{i}{j_1}}{p}) = \cV_{\delta-1}(\jump{\node{i}{j_2}}{p})$. Since port $p$ was arbitrarily chosen, it follows from the recursive definition of the view that $\cV_{\delta}(\node{i}{j_1}) = \cV_{\delta}(\node{i}{j_2})$, and so we have the claim.
\end{proof}

Observe that the nodes $a_l = v_{l}(0)$ and $b_l = v_{l}(2^{l-1})$ have distinct views in $G_l$. Indeed, consider a sequence of $l$ traversals along port $2$, starting from nodes $a_l$ and $b_l$.
We argue, by induction on $i\in\{0,\ldots,l-1\}$, that after $i$ edge traversals the node reached from $a_l$ is $v_{l-i}(0)$, and the node reached from $b_l$ is $v_{l-i}(2^{l-1-i})$.
For $i=0$ the claim is trivial and hence assume that it holds for some $0\leq i<l-1$.
The edge with port number $2$ at $v_{l-i}(0)$ clearly leads to $v_{l-1-i}(0)$ as required.
Hence, it remains to argue that there is an edge between $v_{l-i-1}(2^{l-2-i})$ and $v_{l-i}(2^{l-1-i})$ in $G$.
According to construction of edges between the levels $l-2-i$ and $l-1-i$ in Stage~4, we need to argue that
\begin{equation} \label{eq:ind-dist}
\pi_{l-1-i}(j)=2^{l-1-i},\quad\textup{where }j=2^{l-2-i}.
\end{equation}
By~\eqref{eq:involution},
$$\pi_{l-1-i}(j)=\left(j-2^{l-1-i}b_{l-1-i}(j) - 2^{l-2-i}b_{l-2-i}(j)\right) + 2^{l-1-i}b_{l-2-i}(j) + 2^{l-2-i}b_{l-1-i}(j).$$
We have $b_{l-1-i}(j)=0$ because $2^{l-2-i}\mod 2^{l-i}=0<2^{l-1-i}$, and $b_{l-2-i}(j)=1$ because $2^{l-2-i}\mod 2^{l-1-i}\geq 2^{l-2-i}$.
Thus, $\pi_{l-1-i}(j)=j-2^{l-2-i}+2^{l-1-i}=2^{l-1-i}$ as required, which completes the proof of \eqref{eq:ind-dist}.
Thus, for $i = l-1$, we reach nodes $v_1(0)$ and $v_1(1)$, respectively. Then, after following port $2$ for the $l$-th time, we reach nodes $v_{0}(0)$ and $v_{0}(1)$, respectively. Finally, after following port $2$ for the $(l+1)-$th time, We reach nodes $v_{0}(1)$ and $v_{l+1}(1)$, respectively.

In the last step of the traversal of this sequence of ports, node $v_{0}(1)$ is entered by port $1$, while node $v_{l+1}(1)$ is entered by port $2^l$.
Hence, $\cV(a_l) \neq \cV(b_l)$. On the other hand, $\delta(v_{l}(0), v_{l}(2^{l-1})) = l-1$, so by Lemma~\ref{lem:truncated_depth}, $\cV_{l-1}(a_l) = \cV_{l-1}(b_l)$. We obtain the following claim.

\begin{proposition}\label{pro1}
For any integer $l\geq 6$, there exists a graph $G_l$ on $(l+2)2^l$ nodes, at most $2^{2l}$ edges, and diameter at most $3$, which contains a pair of nodes $a_l$, $b_l$ having distinct views and having the same views up to depth $l-1$.\qed
\end{proposition}

This result completes our proof for the case of graphs of diameter $3$. Now, in order to obtain an asymptotic lower bound of $\Omega(D\log(n/D))$, where $n$ and $D$ are, respectively, the size and the diameter of a graph, we modify each of $G_l$'s to obtain graphs of arbitrarily large diameter.

Let $D$ be an odd integer.
For each $G_l$, $l\geq 1$, define $\subdivide{G_l}{D}$ to be a graph constructed by replacing each edge $\{u,v\}$ from $G_l$ by a path $P(\{u,v\})$ of length $D$ with endpoints $u$ and $v$.
Note that $|V(\subdivide{G_l}{D})|=|V(G_l)|+(D-1)|E(G_l)|$ and $|E(\subdivide{G}{D})|=D|E(G_l)|$.
Also, $\subdivide{G_l}{1}=G_l$.
We define the port labeling $\portlabel_D$ for $\subdivide{G_l}{D}$ as follows.
For each $\{u,v\}\in E(G_l)$ take the corresponding path $P(\{u,v\})=(u,x_1,\ldots,x_{D-1},v)$ and set
$\portlabel_D(u,x_1)=\portlabel(u,v)$,
$\portlabel_D(v,x_{D-1})=\portlabel(v,u)$.
The remaining port labels of $P(\{u,v\})$ are assigned arbitrarily but in such a way that whenever two edges of $G_l$ have the same port labels at the endpoints, then we select isomorphic port labelings for the two corresponding paths in $\subdivide{G_l}{D}$. Formally, for any two edges $\{u,v\}$ and $\{u',v'\}$ of $G_l$ satisfying
$\portlabel(u,v)=\portlabel(u',v')$ and $\portlabel(v,u)=\portlabel(v',u')$, for the two corresponding paths $P(\{u,v\})=(u=x_0,x_1,\ldots,x_{D-1},x_D=v)$ and $P(\{u',v'\})=(u'=x_0',x_1',\ldots,x_{D-1}',x_D'=v')$ it holds that
$\portlabel_D(x_j,x_{j+1})=\portlabel_D(x_j',x_{j+1}')$ and
$\portlabel_D(x_{j+1},x_{j})=\portlabel_D(x_{j+1}',x_{j}')$ for each $j\in\{0,\ldots,D-1\}$.
The latter is possible for any $D$ when $\portlabel(u,v)\neq\portlabel(v,u)$ and it is possible for odd $D$ for `symmetric' edges, i.e., when $\portlabel(u,v)=\portlabel(v,u)$.
As an example of such labeling consider the following.
If $D$ is odd and $\portlabel(u,v)=\portlabel(v,u)$, then we set
$$\portlabel_D(x_j,x_{j-1})=1\textup{ and }\portlabel_D(x_j,x_{j+1})=2\textup{ for each }j\in\{1,\ldots,\lfloor D/2\rfloor\},$$
and 
$$\portlabel_D(x_j,x_{j-1})=2\textup{ and }\portlabel_D(x_j,x_{j+1})=1\textup{ for each }j\in\{\lfloor D/2\rfloor+1,\ldots,D-1\}.$$
If, on the other hand, $\portlabel(u,v)\neq\portlabel(v,u)$, then one can set
$$\portlabel_D(x_j,x_{j-1})=1\textup{ and }\portlabel_D(x_j,x_{j+1})=2\textup{ for each }j\in\{1,\ldots,D-1\}.$$
We also have the following claim.

\begin{claim} \label{claim:new_diam}
For each $l\geq 1$ and $D\geq 1$ it holds that $\diam{\subdivide{G_l}{D}}\leq 3D$.\qed
\end{claim}

We now consider the nodes $a_l, b_l \in V(G_l)$ satisfying Proposition~\ref{pro1}, and characterize their (truncated) views within graph $\subdivide{G_l}{D}$.

\begin{lemma} \label{lem:diff-in-dG_l}
For any $l\geq 1$, $i\leq l-1$, and odd $D\geq 1$, in graph $\subdivide{G_l}{D}$ we have:  $\cV_{Di}(a_l)=\cV_{Di}(b_l)$ and $\cV(a_l)\neq\cV(b_l)$.
\end{lemma}
\begin{proof}
In order to prove that $\cV_{Di}(a_l)=\cV_{Di}(b_l)$, we will use the characterization from Claim~\ref{claim:isomorphic-paths}. Let $P_j=(u_0^j,u_1^j,\ldots,u_{kD}^j)$, $j\in\{1,2\}$, be any two non-backtracking paths in $\subdivide{G_l}{D}$ such that $u_0^1=a_l$ and $u_0^2=b_l$.

By construction, $P_j'=(u_0^j,u_{D}^j,u_{2D}^j,\ldots,u_{kD}^j)$ is a path in $G_l$ for each $j\in\{1,2\}$.
By the definition of port labeling of $\subdivide{G_l}{D}$, for paths ending at nodes within $V(G_l)$, the port labelings of $P_1$ and $P_2$ are identical if and only if the port labelings of $P_1'$ and $P_2'$ are identical. Thus, $P_1$ and $P_2$ are isomorphic in $\subdivide{G_l}{D}$ if and only if $P_1'$ and $P_2'$ are isomorphic in $G_l$. Since $i\leq l-1$, by Claim~\ref{claim:isomorphic-paths} we obtain that $\cV_{Di}(a_l)=\cV_{Di}(b_l)$. The fact that $\cV(a_l)\neq\cV(b_l)$ follows from similar arguments.
\end{proof}

\begin{theorem}
Let $D'\geq 3$ and $n'\geq 1$ be arbitrary integers with $n' \geq D'\cdot 2^{12}/3$. There exists a graph $G$ with at most $n'$ nodes and diameter at most $D'$, which contains two nodes having distinct views which are identical when truncated up to depth $\frac{D'-5}{6}\log_2\frac{n'}{D'} - 0.41D'$.
\end{theorem}
\begin{proof}
Let $D$ be the largest odd integer such that $3D \leq D'$. Note that $D \geq (D'-5)/3$ and $1\leq D \leq D'/3$.
Take $G=\subdivide{G_l}{D}$, $a=a_l$ and $b=b_l$, where $l$ is selected so that $n=|V(G)| \geq n'$. Observe that, by Claim~\ref{claim:diam-of-G_l}, the number of nodes of $G$ satisfies: 
$$n = |V(G_l)|+(D-1)|E(G_l)| < D|E(G_l)| <  D 2^{2l} \leq D' 2^{2l}/3.$$ 
Thus, $n \leq n'$ is satisfied if $D'2^{2l}/3 \leq n'$; we put $l = \lfloor \frac12 \log_2 (3n'/D') \rfloor$.  (Note that $l\geq 6$ by assumption.)

By Lemma~\ref{lem:diff-in-dG_l}, the views of $a_l$ and $b_l$ are different in $G$, but the same when truncated up to depth $D(l-1)$. We have:
\begin{align*}
D(l-1) & \geq \frac{D'-5}{3}\cdot \left(\frac12 \log_2\frac{3n'}{D'} - 2\right) 
       =   \frac{D'-5}{6}\log_2\frac{n'}{D'} + \frac{D'}{3}\left(\frac{1}{2}\log_2 3 - 2\right) - \frac{5}{3}\left(\frac{1}{2}\log_2 3 - 2\right) > \\
       & >  \frac{D'-5}{6}\log_2\frac{n'}{D'} - 0.41D'.
\end{align*}
\end{proof}

\section{Final remarks}

We have shown a tight lower bound of $\Omega(D\log (n/D))$ on the depth to which the views of a pair of nodes of a symmetric anonymous network need to be checked in order to decide if their views in the graph are different. We remark that our problem of view distinction can be generalized in the following two directions:
\begin{itemize}
\item One may consider scenarios in which some information (labels) is also encoded at nodes of the network, and also appears as a node-labeling in the definition of the view. (Such an extended definition of views has appeared, e.g., in the context of leader election in networks where not all identifiers are distinct~\cite{YK99}).
\item One may ask about the depth of the view which suffices not only to distinguish a pair of nodes of the same graph having distinct views, but also any pair of nodes of two arbitrary graphs, which have the same view. (This type of distinction is required in, e.g., in so-called map construction problems~\cite{CDK10}.)
\end{itemize}
Since our lower bound concerns a more restricted scenario, it immediately applies to both of the above cases as well. Formally, when considering a pair of graphs, as $n$ and $D$ we take the maximum order and diameter of the two graphs.

{
At the same time, the techniques used by Hendrickx~\cite{Hendrickx14} to show a corresponding upper bound of $O(D\log_2(n/D))$ for distinguishing a pair of nodes of a connected graph can be adapted to apply to all of the above cases as well, including the scenario of distinguishing a pair of views in two different graphs. Indeed, suppose that there exist a graph $G_1$ on $n_1$ nodes with diameter $D_1$ containing a node $v_1$, and a graph $G_2$ on $n_2$ nodes with diameter $D_2$ containing a node $v_2$, such that nodes $v_1$ and $v_2$ have views in their respective graphs indistinguishable up to some distance $l>1$. Then, one can construct a new connected graph $G$ on $n = n_1+n_2$ nodes with diameter $D \leq D_1 + D_2$, in which there exists a pair of nodes with views indistinguishable also up to distance $l$. To achieve this, denoting by $d$ the degrees of $v_1$ in $G_1$ and of $v_2$ in $G_2$, which are necessarily equal, we form $G$ by taking the disjoint union of graphs $G_1$ and $G_2$, and connecting vertices $v_1$ and $v_2$ by an edge labeled with port $d+1$ at both ends.}

{Thus, we can say that the question  of the necessary depth of view reconstruction with respect to the diameter of a symmetric port-labeled networks has been completely resolved.}

\bibliographystyle{abbrv}
\bibliography{distrib}

\end{document}